\documentclass[sigconf, screen, authorversion]{acmart}

\setcopyright{usgovmixed}

\copyrightyear{2018}
\acmYear{2018}
\acmConference[SPAA'18]{30th ACM Symposium on Parallelism in Algorithms and Architectures}{July 16--18, 2018}{Vienna, Austria}
\acmBooktitle{SPAA '18: 30th ACM Symposium on Parallelism in Algorithms and Architectures, July 16--18, 2018, Vienna, Austria}

\acmArticle{4}
\acmPrice{15.00}
\acmDOI{10.1145/3210377.3210410}
\acmISBN{978-1-4503-5799-9/18/07}

\usepackage{graphicx}
\usepackage{color}
\usepackage{tikz}
\usepackage{url}
\usepackage{amsmath}
\usepackage{amsthm}
\usepackage{hyperref}

\newcommand{\bits}{\ensuremath{\text{bits}}}

\newcommand{\ojas}[1]{\noindent{\bf {\color{red}{\sc Ojas's note:}      #1}}}
\newcommand{\cindy}[1]{\noindent{\bf {\color{blue}{\sc Cindy's note:}      #1}}}
\newcommand{\brad}[1]{\noindent{\bf {\color{green}{\sc Brad's note:}      #1}}}

\renewcommand{\ojas}[1]{}
\renewcommand{\cindy}[1]{}
\renewcommand{\brad}[1]{}

\begin{document}
\title{Constant-Depth and Subcubic-Size Threshold Circuits for Matrix Multiplication}

\author{Ojas Parekh}
\affiliation{%
  \institution{Sandia National Laboratories}
  \city{Albuquerque}
  \state{NM}
  \country{USA}
}
\email{odparek@sandia.gov}

\author{Cynthia A. Phillips}
\affiliation{%
  \institution{Sandia National Laboratories}
  \city{Albuquerque}
  \state{NM}
  \country{USA}
}
\email{caphill@sandia.gov}

\author{Conrad D. James}
\affiliation{%
  \institution{Sandia National Laboratories}
  \city{Albuquerque}
  \state{NM}
  \country{USA}
}
\email{cdjame@sandia.gov}

\author{James B. Aimone}
\affiliation{%
  \institution{Sandia National Laboratories}
  \city{Albuquerque}
  \state{NM}
  \country{USA}
}
\email{jbaimon@sandia.gov}

\date{}


\begin{abstract}
Boolean circuits of McCulloch-Pitts threshold gates are a classic model of neural computation studied heavily in the late 20th century as a model of general computation. Recent advances in large-scale neural computing hardware has made their practical implementation a near-term possibility.  We describe a theoretical approach for multiplying two $N$ by $N$ matrices that integrates threshold gate logic with conventional fast matrix multiplication algorithms, that perform $O(N^\omega)$ arithmetic operations for a positive constant $\omega < 3$.  Our approach converts such a fast matrix multiplication algorithm into a constant-depth threshold circuit with approximately $O(N^\omega)$ gates.  Prior to our work, it was not known whether the $\Theta(N^3)$-gate barrier for matrix multiplication was surmountable by constant-depth threshold circuits.  

Dense matrix multiplication is a core operation in convolutional neural network training. Performing this work on a neural architecture instead of off-loading it to a GPU may be an appealing option. 
\end{abstract}

\maketitle

\section{Introduction}
\label{sec:intro}
Neuromorphic computing devices (NCDs) are composed of massively parallel networks of artificial neurons that compute boolean functions, akin to conventional logic gates but with relatively larger fan-in. The simplicity of each neuron affords NCDs immense energy efficiency and scalability, circumventing some of the data movement and energy bottlenecks associated with traditional HPC~\cite{schuman_survey_2017}. 
Within the last few years, large-scale commercial NCD platforms, such as SpiNNaker, IBM's True North, and Intel's Loihi, have emerged and currently offer up to hundreds of millions of neurons, with systems with billions of neurons a near-term likelihood~\cite{sugiarto_high_2016,schuman_survey_2017,davies_loihi:_2018}. Such NCDs were envisioned, for example, to accelerate deep learning; however, a specific application for which NCDs offer a clear and rigorously established advantage over conventional parallel computing remains elusive. Developing theoretical models for NCDs and identifying potential advantages and tradeoffs over other parallel models of computation remain largely open problems.  We consider a well-known theoretical model that captures some of the features offered by NCDs, and we demonstrate resource-efficient constant-time algorithms for a fundamental problem, matrix multiplication. 

We study the computational power of boolean circuits where the fundamental gates have unbounded fan-in and compute a linear threshold function\footnote{As we explain below, for convolutional neural networks, we can restrict the fan-in to a constant, or whatever practical fan-in the architecture supports.}.  Such circuits are rooted in the classical McCulloch-Pitts neuronal model~\cite{mcculloch_logical_1943}, with linear threshold functions serving as plausible models of spiking neurons.  A boolean \emph{threshold gate} with $m$ binary inputs $y_1, y_2, \ldots, y_m$ computes a \emph{linear threshold function}, outputting $1$ if and only if $\sum_{i=1}^m w_iy_i \ge t$, where the integer weights $w_i$ and integer threshold $t$ are constants associated with the gate.  Rational $w_i$ and $t$ may be represented, for example, by multiplying such $w_i$ and $t$ with a sufficiently large number. There are several natural measures of complexity associated with boolean circuits, including \emph{size}: the total number of gates, \emph{depth}: the length of the longest directed path from an input node to an output node, \emph{edges}: the total number of connections between gates, and \emph{fan-in}: the maximum number of inputs to any gate.

Threshold gates are the most basic model for an artificial neuron, and as such, all currently available neuromorphic computing architectures provide a hardware implementation of threshold circuits.  We consider threshold circuits with constant depth and polynomial size with respect to the total number of inputs; this class of circuits is called $TC^0$.  Such circuits represent a plausible model of constant-time parallel computing. This is a notion of perfect parallelizability, faster than the polylogarithmic time allowed in the complexity class $NC$. $TC^0$ circuits can compute a variety of functions including integer arithmetic, sorting, and matrix multiplication~\cite{sima_general-purpose_2003,siu_discrete_1995}.  There is also a $TC^0$ threshold-gate circuit of sublinear size to compute the parity of $n$ bits~\cite{siu_depth-size_1991}. In contrast, any constant-depth circuit with NOT gates and unbounded-fan-in AND and OR gates requires a superpolynomial number of gates~\cite{furst_parity_1984,yao_separating_1985}.

Understanding the power and limitations of $TC^0$ circuits has been a major research challenge over the last couple of decades.  The 1990's saw a flurry of results showing what $TC^0$ circuits could do~\cite{sima_general-purpose_2003,siu_discrete_1995}, while more recent results have focused on lower bounds showing what $TC^0$ circuits cannot do~\cite{kane_super-linear_2015}.         
$TC^0$ has been studied as a theoretical model. Its practicality is an open question. Currently, large-scale electronic circuits with high fan-in may be difficult to implement. However, neural-inspired architectures may offer hope. 

The adult human brain contains about 100 billion neurons, with maximum fan-in 10,000 or larger some regions, operating at under 50 watts~\cite{azevedo_equal_2009}. Though impressive, this figure represents a single class of instance size, so one might wonder how a synthetic system based on the physical limitations governing a brain might scale asymptotically.  We are not aware of any generative brain models for which this has been analyzed. However, a fan-in that grows with the total system size seems plausible for a 3-dimensional circuit such as the brain.  Although large fan-in is a critical resource in designing $TC^0$ algorithms, in practical neuromorphic devices resource requirements may grow as a function of fan-in.  For example, the available numerical precision or dynamic range may decrease as fan-in increases, resulting in an overall increase in energy expenditure or execution time.  Thus constant depth, in the $TC^0$ sense, may not practically equate to constant time. However, ideal theoretical algorithms may still guide the development of resource-efficient practical algorithms as neuromorphic architectures become more prevalent.

While neuromorphic computing has long focused on the use of analog computation to emulate neuronal dynamics~\cite{indiveri2011neuromorphicMoreAuthors}, recent years have seen rapid development of novel digital CMOS neural hardware platforms which can scale to very large numbers of neurons~\cite{khan2008spinnaker,merolla2014million}. While initially designed for large biologically inspired circuits, these architectures are attracting attention as an alternative to conventional CMOS architectures for accelerating machine learning algorithms such as deep artificial neural networks~\cite{esser2015backpropagation}. 
Many of these neural architectures, such as TrueNorth and the SpiNNaker platform, achieve considerable benefits in energy and speed by using large numbers of simple digital spiking neurons instead of a relatively smaller number of powerful multi-purpose processors.  These systems are almost configurable threshold gate circuits, except that they are capable of extended temporal dynamics.
Scientific computing is an application domain for which neural architectures are often quickly dismissed. There is a perception that human cognition is better for data-centric functions, such as image recognition, and for abstract decision making than for precise numerical calculations, particularly at large scale.   While biologically-inspired neural algorithms are often probabilistic or approximate, the neuronal-level computations in large scale neural architectures are sufficiently precise for numerical computation. \cindy{Does this need a reference?}

We consider a fundamental scientific-computing-inspired problem: can one produce constant-depth threshold circuits that compute the product of two $N \times N$ matrices using $O(N^{3-\varepsilon})$ gates for constant $\varepsilon > 0$?  For matrices with relatively large entries (say $\Omega(N)$ bits), this goal seems out of reach as we would need to output $\Omega(N^3)$ bits in the worst case.  However, prior to our work, it was not known if this was possible even for \emph{binary} matrices, those with entries that are all either 0 or 1. 

We show how to multiply two $N \times N$ integer matrices with $O(\log N)$-bit entries using $O(N^{3-\varepsilon})$ gates in constant depth.  The n{\"a}ive algorithm based on the definition of matrix multiplication requires $\Theta(N^3)$ arithmetic operations.  Our results are based on classical breakthroughs for fast matrix multiplication~\cite{Strassen1969}: multiplying two $N \times N$ matrices using $O(N^\omega)$ arithmetic operations, for a positive constant $\omega < 3$, that depends on the particular fast matrix multiplication being employed.   These techniques can be extended, in a relatively straightforward manner, to yield $O(\log N)$-time conventional parallel algorithms (for architectures such as PRAMs) with $O(N^\omega)$ total work.  In contrast, we give a constant-time algorithm, in the threshold circuit model, with approximately $O(N^\omega)$ total gates, which is a reasonable measure of total work.  

\ojas{Perhaps we should go ahead and mention: In fact we can construct circuits using $\widetilde{O}(N^{\omega})$ gates, given any conventional fast matrix multiplication algorithm performing $O(N^{\omega})$ multiplications, or maybe mention at a higher level that our gate complexity matches the multiplication complexity of a fast MM algorithm}

One of our motivations for neural-circuit-based matrix multiplication is convolutional neural networks for deep learning. See Section~\ref{sec:MM-apps} for more details. Deep learning is a major driver for neural-inspired architectures. A current vision for using these architectures for deep learning requires the matrix multiplication to be moved off-system to a GPU.  If circuit-based matrix multiplication can be made practical, perhaps this computation can be left on-chip, avoiding energy-intensive and slow I/O. 

We also consider the somewhat simpler problem of determining whether the trace of $A^3$ is at least $\tau$, for an $N \times N$ integer matrix $A$ with entries of size $O(\log N)$ bits.  This case illustrates the fundamental ideas of our approach and has applications in social network analysis, particularly to triangle counting.   The problem we solve allows one to answer: ``Does a graph $G$ have at least $\tau$ triangles?''  The user may select a relevant value of $\tau$. See Section~\ref{sec:MM-apps} for more details on triangles, social network analysis, and picking $\tau$. 

There is a simple depth-$2$ threshold circuit to solve this problem for a graph $G=(V,E)$.  The circuit has an input variable, $x_{ij}$ for $i,j \in V$ with $i < j$; the variable $x_{ij}$ is $1$ if $ij \in E$ and $0$ otherwise.  The first layer of the circuit consists of a gate, $g_{ijk}$ for each triple $i,j,k \in V$ with $i < j < k$.  The gate $g_{ijk}$ computes the value of the linear threshold function $x_{ij} + x_{ik} + x_{jk} \geq 3$ as an output $y_{ijk}$.  That is, the gate fires ($y_{ijk} = 1$) if and only if all edges in the triangle on $i$, $j$, and $k$ are in the graph.  The second layer consists of a single output gate that computes the linear threshold function $\sum_{i,j,k \in V: i < j < k} y_{ijk} \geq \tau$; this gate fires if and only if the number of triangles in $G$ is at least $\tau$.  The circuit has ${N \choose 3} + 1 = \Theta(N^3)$ gates.

We ask (and answer) whether it is possible to beat the size of this threshold circuit in constant depth.  This is akin to asking if it is possible to beat the n{\"a}ive matrix multiplication algorithm with an algorithm that performs $O(N^\omega)$ operations for $\omega < 3$.  In fact the above threshold circuit is a specialization of the n{\"a}ive matrix multiplication algorithm.

The analysis of our new threshold circuits is more involved than analyzing conventional fast matrix multiplication methods. We must explicitly consider sparsity (see Definition~\ref{def:d-sparsity}), a measure of how many times a matrix element or intermediate result is part of a computation during the fast multiplication. Thus, while we use existing fast matrix multiplication techniques to achieve our results, we use them in a new context.  Our performance results exploit different features of fast matrix multiplication techniques than those traditionally used.


\subsection*{Results and contributions}
Consider a fast recursive or divide-and-conquer matrix multiplication algorithm like Strassen's, with run-time complexity $O(N^{\omega})$.  We consistently use $\omega$ as the exponent in the runtime complexity of a base non-circuit fast matrix multiplication algorithm.  Our results leverage such a fast matrix multiplication to construct a constant-depth threshold circuit with $\widetilde{O}(N^{\omega + \varepsilon})$-gates, where $\varepsilon$ depends on the depth of the circuit. 

Our main result is an $O(d)$-depth, $\widetilde{O}(N^{\omega + O(\gamma^d)})$-gate threshold circuit for multiplying two $N \times N$ integer matrices with $O(\log N)$-bit entries, for a positive integer $d$ and constant $\gamma < 1$. Specifically, for a given integer $d$, the depth is $4d+1$.  The constant $d$ is a multiplicative factor hidden in the $\widetilde{O}$ for the number of gates. Section~\ref{sec:trace} gives a more detailed discussion of the value of $\gamma$.  For Strassen's algorithm it is about $0.491$. The constant multiplier of $\gamma^d$ is about $1.581$ for Strassen's algorithm. Thus for $d > 3$, this circuit will have $O(N^{3-\varepsilon})$ gates for positive constant $\varepsilon > 0$. 
We also give a $O(\log \log N)$-depth, $\widetilde{O}(N^{\omega})$-gate circuit for this task.

We present a simplified circuit of the same gate complexity and slightly lower depth ($2d+2$) for computing the trace of $A^3$, for an $N \times N$ integer matrix $A$.  This gives triangle counts for a graph $G$ with adjacency matrix $A$ (see Section~\ref{sec:problem-statement}).
Our circuits implement limited-depth versions of fast divide-and-conquer matrix multiplication, and our techniques should extend to other types of algebraic divide-and-conquer algorithms.

Our contributions are:
\begin{itemize}
\item This work revives and redirects research on designing algorithms for a classical theoretical model of parallel computation to a data science problem on an emerging class of neural-inspired parallel architectures. We show that threshold circuits, comprised of threshold gates that model neurons, might be applicable to linear-algebra computations for deep learning on new neuromorphic hardware. 
\item We give $O(\log \log n)$-depth and $\widetilde{O}(N^\omega)$-gate threshold circuits for computing the product of two $N\times N$ matrices with $O(\log N)$-bit entries, where $N^{\omega}$ is the complexity of a fast conventional matrix multiplication algorithm like Strassen's.
\item We give constant-depth threshold circuits with $\widetilde{O}(N^{\omega + \varepsilon})$ gates, where $\varepsilon$ is exponentially small in the depth of the circuit.  
\item Our circuits are elementary and are composed entirely of copies of a relatively simple depth-2 threshold circuit that performs addition.  We hope this will facilitate implementation of our approach in neural-inspired hardware.
\end{itemize}

\section{Preliminaries and problem statement}

\subsection{Fast matrix multiplication algorithms}\label{sec:fast-matmul}

Strassen developed the first matrix multiplication algorithm requiring $O(N^{3-\varepsilon})$ multiplications~\cite{Strassen1969}.  Strassen observed that one can compute the matrix product, $C = AB$ for $2 \times 2$ matrices $A$ and $B$ using 7 multiplications rather than the 8 multiplications required by the n{\"a}ive algorithm. The reduction in multiplications comes at the expense of additional additions and subtractions. 

\begin{figure}
\centering
\begin{align*}
\begin{aligned}
M_1 &= A_{11}(B_{12} - B_{22})\\
M_2 &= (A_{21} + A_{22})B_{11}\\
M_3 &= (A_{11} + A_{22})(B_{11} + B_{22})\\
M_4 &= A_{22}(B_{21} - B_{11})\\
M_5 &= (A_{11} + A_{12})B_{22}\\
M_6 &= (A_{21} - A_{11})(B_{11} + B_{12})\\
M_7 &= (A_{12} - A_{22})(B_{21} + B_{22}).
\end{aligned}
\qquad
\begin{aligned}
C_{11} &= M_3 + M_4 - M_5 + M_7\\
C_{12} &= M_1 + M_5\\
C_{21} &= M_2 + M_4\\
C_{22} &= M_1 - M_2 + M_3 + M_6.
\end{aligned}
\end{align*}
\caption{Strassen's algorithm for multiplying two $2 \times 2$ matrices $A$ and $B$. The 7 multiplications computed in Strassen's algorithm are represented by $M_1,\ldots,M_7$. Each $M_i$ is the product of weighted sums of entries of matrices $A$ and $B$.  The entries of product matrix $C$ are then computed from the $M_i$ using only addition and subtraction. One can verify by substitution and expansion that the entries of $C$ are set to the proper expressions involving entries of $A$ and $B$.}
\label{fig:strassen}
\end{figure}

Figure~\ref{fig:strassen} gives Strassen's algorithm for $2 \times 2$ matrices.  The algorithm is generalized to $N \times N$ matrices $A$ and $B$, where $N = 2^l$ for some positive integer $l$, as follows.  We partition $A$ and $B$ into 4 blocks, each of size $N/2 \times N/2$, and let $A_{ij}$ and $B_{ij}$ refer to these blocks, for $i,j \in \{1,2\}$.  The above equations remain correct. However, each $M_i$ now represents a multiplication of two $N/2 \times N/2$ matrices.  We can recursively apply the above procedure to perform each of these multiplications until the blocks are individual matrix elements or, for more practical applications, sufficiently small.  For each of the $l = \log_2 N$ levels of recursion, we invoke 7 recursive matrix multiplications, resulting in a total of $7^{\log_2 N} = N^{\log_2 7} \approx N^{2.81}$ scalar multiplications.  The recurrence relation for the total number of arithmetic operations is $T(N) = 7 \cdot T(N/2) + 18 \cdot (N/2)^2$ and $T(1) = O(1)$. The $18 \cdot (N/2)^2$ term arises from the 18 additions or subtractions on $N/2 \times N/2$ blocks in the expressions above. This recurrence shows the total number of scalar additions or subtractions is also $O(N^{\log_2 7})$. 

Although Strassen's seminal approach was based on a fast matrix multiplication algorithm for $2 \times 2$ matrices, subsequent work has yielded improved algorithms by employing a fast matrix multiplication algorithm involving larger square matrices as well as more sophisticated techniques.  The currently best known algorithm requires $O(N^{2.373})$ operations~\cite{le_gall_powers_2014}.  For the sake of exposition, we view fast matrix multiplication algorithms as recursive divide-and-conquer approaches, yet our techniques extend to the more general tensor perspective of fast matrix multiplication.  See the survey by Bl{\"a}ser~\cite{blaser_fast_2013} for a detailed introduction to and history of this area, including the connection between the (border) rank of the matrix multiplication tensor and fast matrix multiplication algorithms. 

\subsection{Technical challenges}
The divide-and-conquer Strassen's algorithm has a natural $O(\log N)$-time parallel (PRAM) implementation with a total work of $O(N^{\log_2 7})$ arithmetic operations.  The main question we consider is whether Strassen's approach and subsequent improvements of it can yield a $\emph{constant-time}$ algorithm implemented using threshold circuits with $O(N^{3-\varepsilon})$ gates, where the latter is a measure of total work.  The recursive ($O(\log N)$-depth) implementation of Strassen's algorithm only performs scalar multiplications during the base case. However, it performs matrix additions and subtractions at each level before and after the recursion, reusing computed results.  

If we attempt to implement Strassen's approach without recursion and the consequent reuse of computed results, we must first compute $N^{\log_2 7}$ scalars that are linear combinations of entries of $A$ and another $N^{\log_2 7}$ scalars representing linear combinations of entries of $B$.  The main technical hurdle is that such linear combinations involve up to $N$ entries of $A$ or $B$, and we seek to compute $O(N^{\log_2 7})$ of these sums with constant-depth threshold circuits.  A n\"aive implementation would require $\Omega(N^{1+\log_2 7})$ gates.

We overcome this hurdle by selecting a constant or $O(\log \log N)$ number of levels of recursion out of the $O(\log N)$ levels suggested by the standard implementation of Strassen's approach.  We make the notion of selecting a level of recursion precise in Section~\ref{sec:trace-circuit}.    This allows us enough reuse of computed results within the confines of a constant-depth or $O(\log \log N)$-depth circuit to achieve our results.  We note that we must carefully select the levels of recursion employed; for instance, simply selecting every $k$th level does not achieve our best results.

\subsection{Problem statement}
\label{sec:problem-statement}

We develop threshold circuits to compute the matrix product $C = AB$ of two $N \times N$ integer matrices $A$ and $B$.  Our results assume the entries of $A$ and $B$ require at most $O(\log N)$ bits.  We also consider a related problem: given an integer matrix $A$ as above and an integer $\tau$, determine whether the matrix trace of $A^3$ is at least $\tau$.  This problem is solved by a simpler threshold circuit than for computing matrix product and serves to illustrate our main ideas.  It also has applications to triangle counting in graphs as we describe below.

Let $A$ be the $N \times N$ symmetric adjacency matrix of a simple graph $G = (V,E)$ with $N$ nodes: for $i,j \in V$, $A_{ij} = A_{ji} = 1$ if $ij \in E$, and $A_{ij} = A_{ji} = 0$ otherwise.  Since there are no self-loops in the graph, we have $A_{ii} = 0$ for all nodes $i$. Consider the square of the adjacency matrix, $C = A^2$.  For $i,j \in V$ with $i \not =j$, $C_{ij} = \sum_{k \in V} A_{ik}A_{kj} = | \{ k \in V \mid ik \in E \text{ and } kj \in E \text{ and } k\neq i,j\}|$, which is the number of paths of length 2 between $i$ and $j$.  If there is an edge between the nodes $i$ and $j$, then each path of length 2 between them, along with the edge $ij$, forms a triangle in $G$.  Moreover, every triangle containing $i$ and $j$ arises in this way.  Suppose $G$ has $\Delta$ triangles.  Then,
\begin{equation}\label{eq:triangle-count}
3\Delta = \sum_{\substack{i,j \in V: i < j}} A_{ij} C_{ij},
\end{equation}
since the sum counts each triangle once for each of its edges.  Thus one can count the triangles in $G$ by summing some of the entries of $A^2$.  An equivalent computation is the trace of $A^3$, $\text{trace}(A^3)$, which, from \eqref{eq:triangle-count}, is equal to $6\Delta$. This counts the loop from each vertex in each direction.

We employ a threshold circuit implementation of fast matrix multiplication algorithms to compute this sum in constant depth using $O(N^{3-\varepsilon})$ gates.  In fact the exponent of our gate count can be made arbitrarily close to the exponent of the arithmetic operation count for the best possible fast matrix multiplication algorithm.

We explain our notion of a fast matrix multiplication algorithm.  We assume we are given an algorithm for multiplying two $T \times T$ matrices using a total of $r$ multiplications (for Strassen's algorithm, $T = 2$ and $r = 7$).  We assume $N = T^l$ for some positive integer $l$.  As outlined in Section~\ref{sec:fast-matmul}, this yields a recursive algorithm for computing the product of two $N \times N$ matrices, $C = AB$, using a total of $r^l = r^{\log_T N} = N^{\log_T r}$ scalar multiplications.

As with Strassen's algorithm, we assume we are given a list of $r$ expressions for each of the multiplications, $M_1,\ldots,M_r$; we view each $M_i$ as an expression involving the $T^2$ different $N/T \times N/T$ blocks of $A$ and $B$.  In particular each $M_i$ is a product of a $\{-1,1\}$-weighted sum of blocks of $A$ with a $\{-1,1\}$-weighted sum of blocks of $B$.  We also assume the fast matrix multiplication algorithm provides a list of $T^2$ expressions, each representing a $N/T \times N/T$ block of $C$ as a $\{-1,1\}$-weighted sum of the $M_i$.  More general fast matrix multiplication algorithms may allow the $M_i$ to be products of linear combinations with rational weights beyond $\{-1,1\}$ (likewise for the entries of $C$).  Although we do not present details here, our techniques can be extended for such fast matrix multiplication algorithms (these weights correspond to the $w_i$ in Lemma~\ref{lem:stronger-sum-circuit} in Section~\ref{sec:tc0-arithmetic-circuits}). 

For $1 \leq i \leq r$, let $a_i$ be the number of distinct blocks of $A$ that appear in the expression $M_i$, and let $b_i$ be defined analogously with respect to $B$.  We let $c_i$ be the number of expressions for blocks of $C$ in which $M_i$ appears.
\begin{definition}\label{def:d-sparsity}We let
\begin{align*}
s_A = \sum_{1 \leq i \leq r} a_i,\ s_B = \sum_{1 \leq i \leq r} b_i,\ \text{and } s_C = \sum_{1 \leq i \leq r} c_i.
\end{align*}
We define the \emph{sparsity} of a fast matrix multiplication algorithm as $s = \max\{s_A,s_B,s_C\}$.
\end{definition}
Sparsity will be an essential ingredient of our analysis, and we better motivate it in Section~\ref{sec:trace}. Others~\cite{ballard_improving_2016,Bini_Lotti_1980} consider sparsity in analyzing and improving the numerical stability of fast matrix multiplication algorithms, though they do not refer to it by the same name. 
\cindy{check: do they call this sparsity? I pulled the paper and the word sparsity doesn't appear to be in the paper. Maybe we should add the phrase ``but they don't call it sparsity.''  It would be better to include their term, but it was hard to find in a quick look.  Maybe they just give it a notation.} 

We use the following notation in proofs of circuit quality. We define $\bits(m)$ as the minimum number of bits required to express the nonnegative integer $m$ in binary, i.e.~the least integer $l$ such that $m < 2^l$.  

\section{Basic $TC^0$ arithmetic circuits}
\label{sec:tc0-arithmetic-circuits}

We first develop the fundamental $TC^0$ arithmetic circuits on which our results rely.  Our circuits are designed with neuromorphic implementation in mind, and we try to favor simple constructions over those that offer the least depth or gate count.  The bulk of the computation performed by our circuits is computing the bits of integer-weighted sums of nonnegative integers, $\sum_i w_i x_i$, where the nonnegative $x_i$ depend upon the inputs to the circuit but the weights $w_i$ are constants associated with the circuit.  



\ojas{State edge and fan-in complexity too: edge should be linear in gates; using siu et al. would result in superlinear edges w.r.t. gates}
Our first circuit follows from a classical technique to compute symmetric functions in $TC^0$ by Muroga from 1959~\cite{Mur59,Min61}; it is also a special case of a more general result by Siu et al.~\cite{siu_depth-size_1991}.  We include a proof to demonstrate the simplicity of the construction.
\begin{lemma}\label{lem:sum-circuit}
Let $s = \sum_i w_i x_i$ be an integer-weighted sum of bits, $x_i \in \{0,1\}$.  We assume $s \geq 0$ and fix an integer $l$ such that $s \in [0,2^l)$.  For $1 \leq k \leq l$, the $k$th most significant bit of $s$ can be computed by a depth-2 threshold circuit using $2^k + 1$ gates.
\end{lemma}
\begin{proof} We define bool$(P)$, for a predicate $P$, to be $1$ if predicate $P$ is true and $0$ otherwise.

The $k$th most significant bit of $s$ is 1 precisely when $s$ lies in one of the intervals $[i2^{l-k},(i+1)2^{l-k})$ for some odd integer $1 \leq i < 2^k$.  The interval enumerates over all combinations of bits less signicant than the $k$th and the odd multipliers $i$ enumerate over all combinations of bits more significant than $k$. The first layer of our circuit computes the function $y_i := \text{bool}(s \geq i2^{l-k})$, for $1 \leq i \leq 2^k$.  The output of the circuit is $\text{bool}(\sum_{i \text{ odd}} (y_i - y_{i+1}) \geq 1)$, since $y_i - y_{i+1}$ is 1 if $s \in [i2^{l-k},(i+1)2^{l-k})$ and is 0 otherwise. 

\end{proof}

The circuit construction for the above lemma requires an integer $l$ such that the sum $s$ is guaranteed to be in $[0,2^l)$.  Note that if $s \notin [0,2^l)$, the circuit for any $k$ outputs 0.  We build upon the above to obtain our primary addition circuit.  The next lemma is a generalized and strengthened version of Siu et al.'s Lemma~\ref{lem:sum-circuit} in \cite{siu_depth-size_1991} for the depth-2 case.
\cindy{I looked quickly at this Lemma in Siu et al but it requires some notation from earlier in the paper.  So I will check this a little later.}
\begin{lemma}\label{lem:stronger-sum-circuit}Let $s = \sum_i w_i z_i$ be an integer-weighted sum of $n$ nonnegative numbers $z_i$, each with at most $b$ bits.  We assume $s \geq 0$ and let $w = \max_i |w_i|$.  The sum $s$ can be computed by a depth-$2$ threshold circuit with $O(wbn)$ gates.
\end{lemma}
\begin{proof}
The sum $s$ requires at most $\bits(nw(2^b-1)) \leq \bits(n) + \bits(w) + b$ bits. (Recall the defintion of $\bits()$ in Section~\ref{sec:problem-statement}).

First we compute the $j$th (least significant) bit of $s$, for $1 \leq j \leq b$.  Let $s_j = \sum_i w_i \tilde{z_i}$, where $\tilde{z_i}$ is obtained from $z_i$ by ignoring all but the least significant $j$ bits.  Note that $s_j$ requires at most $\bits(n) + \bits(w) + j$ bits and that $s$ and $s_j$ have the same least significant $j$ bits.  Furthermore, since we are given the bits of each $\tilde{z_i}$, we may treat $s_j$ as an integer-weighted sum of bits, where each weight is a product of some $w_i$ with a power of 2.  Thus we may compute the $jth$ bit of $s_j$ by appealing to Lemma~\ref{lem:sum-circuit} on $s_j$ with $k = \bits(n) + \bits(w) + 1$.  To see this recall that $k$ represents the $k$th most significant bit. The $j$th least significant bit of $s_j$ is the $(\bits(s_j) - j + 1) = n + w + j - j + 1$ most significant bit.   \cindy{I was initially a little nervous that the transformation of the $\tilde{z_i}$ to bits blows up $n$, the number of terms by a factor of $\log j$, but the number of terms (numbers) doesn't matter for the number of gates in Lemma~\ref{lem:sum-circuit}. It's the number of bits in the representation of the sum, and that will not be changed by the representation. We probably don't need to add anything to the proof, but I'm documenting anything that makes me stop and think, in case it becomes relevant in the ``polishing'' phase.}  This circuit requires $2^k+1 = 2\cdot 2^{\bits(n)}\cdot 2^{\bits(w)} +1 = O(wn)$ gates, hence $O(bwn)$ gates suffice to compute the $b$ least significant bits of $s$.

Appealing to Lemma~\ref{lem:sum-circuit} to compute each of the remaining $a = \bits(n) + \bits(w)$ most significant bits of $s$ requires $O(\sum_{k=1}^a 2^k) = O(2^a) = O(wn)$ gates.  This is improved in practice by observing that the functions $y_i$ computed for $k=\bits(n) + \bits(w)$ in the proof of Lemma~\ref{lem:sum-circuit} include those required for all the most significant $\bits(n) + \bits(w)$ bits of $s$.  
\end{proof}

We need to compute products of numbers as well.  However, the products we compute are only used as inputs to other threshold gates, and we do not need an explicit base-2 representation of the bits of these products.  A more generic representation suffices: a \emph{representation} of an integer $x$ is an integer-weighted sum of binary variables, $x = \sum_{1 \leq i \leq d} w_i x_i$ with $x_i \in \{0,1\}$ and $d$ polynomial in $\bits(x)$.

\begin{lemma}\label{lem:prod-circuit}
A representation of the product of three $m$-bit nonnegative integers can be computed by a depth-1 threshold circuit with $m^3$ gates.  
\end{lemma}
\begin{proof}We compute a representation of the product of $x = \sum_{1 \leq i \leq m} 2^{i-1} x_i$, $y = \sum_{1 \leq j \leq m} 2^{j-1} y_j$, and $z = \sum_{1 \leq k \leq m} 2^{k-1} z_k$, with $x_i,y_j,z_k \in \{0,1\}$.  Thus $xyz = \sum_{1 \leq i,j,k \leq m} 2^{i+j+k-3}x_iy_jz_k$, which is a representation for $xyz$.  This representation differs from the standard binary representation in that $2^{i+j+k-3}$ can represent the same power of 2 for different values of $i, j$, and $k$.  We use $m^3$ gates in a single layer with predicates $x_i + y_j + z_k \geq 3$ to compute $x_i y_j z_k \in \{0,1\}$ for $1 \leq i,j,k \leq m$.  
\end{proof}
Our results require only the above relatively simple arithmetic circuits.  This facilitates practical implementation.

\paragraph*{Negative numbers} The above lemmas give circuits to compute products and weighted sums of nonnegative integers. However, they can be extended to handle negative integers.  We represent each integer $x$ as $x = x^+ - x^-$, where $x^+$ and $x^-$ are each nonnegative.  Other more efficient representations are possible, but we select this one as it makes for a relatively simple presentation and implementation at the cost of a constant-factor overhead in gate and wire count. 

The workhorse subroutine of our circuits, captured by Lemma~\ref{lem:stronger-sum-circuit}, is computing the bits of integer-weighted sums of nonnegative integers, $s = \sum_i w_i x_i$, where the $x_i=x_i^+-x_i^-$ depend upon the inputs to the circuit but the $w_i$ are constant with respect to the inputs.  Let $W^+$ be the set of indices with $w_i > 0$, and let $W^-$ be those indices with $w_i < 0$.  We define $s^+ = \sum_{i \in W^+} w_i x_i^+ + \sum_{i \in W^-} (-w_i)x_i^-$ to be the positive terms in sum $s$, and we define $s^- = \sum_{i \in W^+} w_i x_i^- + \sum_{i \in W^-} (-w_i)x_i^+$ to be the negation of the negative terms in sum $s$.  We have $s = s^+ - s^-$ and $s^+, s^- \geq 0$. Moreover, each of $s^+$ and $s^-$ is an integer-weighted sum of nonnegative integers, hence the bits of each of $s^+$ and $s^-$ may be computed using two separate instances of the circuit of Lemma~\ref{lem:stronger-sum-circuit}.  Each of these circuit instances only depends on the $x_i^+$ and $x_i^-$ hence we may apply them in parallel without increasing the depth of the resulting overall circuit. 



Computing products also incurs extra constant overhead.  The representation of the product of three numbers, $xyz$ described in the proof of Lemma~\ref{lem:prod-circuit} becomes $xyz = \sum_{1 \leq i,j,k \leq m} 2^{i+j+k-3}(x_i^+ - x_i^-)(y_j^+ - y_j^-)(z_k^+ - z_k^-)$.   Thus $\sum_{1 \leq i,j,k \leq m} [2^{i+j+k-3} (x_i^+y_j^+z_k^+ + x_i^+y_j^-z_k^- + x_i^-y_j^+z_k^- + x_i^-y_j^-z_k^+) + (-2^{i+j+k-3})(x_i^+y_j^+z_k^- + x_i^+y_j^-z_k^+  + x_i^-y_j^+z_k^+ + x_i^-y_j^-z_k^-)]$ is also a representation of $xyz$ that requires eight times as many gates to compute, which is still $O(m^3)$.

For ease of exposition, we proceed as if we were only computing positive quantities.  From this point on,  we assume that a number $x$ requires at most $b$ bits, by which we mean that each of $x^+$ and $x^-$ requires at most $b$ bits.

\section{Subcubic $TC^0$ circuits for trace and matrix multiplication}
\label{sec:trace-circuit}

\subsection{Overview}

Our circuits for matrix trace and matrix multiplication implement a given conventional fast matrix multiplication algorithm in both a depth-efficient and gate-efficient manner.  We assume we are given $N \times N$ integer matrices $A$ and $B$ with entries of size $O(\log N)$.  We consider two problems: (1) determining whether $\text{trace}(A^3) \geq \tau$, for an integer $\tau$, and (2) computing the bits of the matrix product $C=AB$.  We consider the first problem in this section, and the second problem in the next section.  

We define trees $\mathcal{T}_A$ and $\mathcal{T}_B$ for the input matrices $A$ and $B$, respectively, based on the recursive or divide-and-conquer structure of the fast matrix multiplication algorithm.  The nodes in $\mathcal{T}_A$ represent weighted sums of blocks of $A$ and likewise for $\mathcal{T}_B$.  The root of $\mathcal{T}_A$ represents the matrix $A$, while the leaves represent weighted sums of its entries. See Figure~\ref{fig:T_A} for a detailed explanation.

\begin{figure}
\centering
\scalebox{0.64}{
\begin{tikzpicture}[level/.style={sibling distance=60mm/#1}]
\tikzstyle{tnode} = [rectangle, draw, rounded corners]
\node [tnode] (z){$A$}
  child {node [tnode] (a) {$A_{11}$}
    child {node [tnode] (b) {$(A_{11})_{11}$}
      child {node (l1) {$\vdots$}
      } 
      child {node (l2) {$\vdots$}}
    }
    child {node [tnode] (g) {$(A_{11})_{12} - (A_{11})_{22}$}
      child {node (l3) {$\vdots$}}
      child {node (l4) {$\vdots$}}
    }
  }
  child {node [tnode] (j) {$A_{12}-A_{22}$}
    child {node [tnode] (k) {$(A_{12}-A_{22})_{11}$}
      child {node (l5) {$\vdots$}}
      child {node (l6) {$\vdots$}}
    }
  child {node [tnode] (l) {$\substack{(A_{12}-A_{22})_{12}\\-(A_{12}-A_{22})_{22}}$}
    child {node (l7) {$\vdots$}}
    child {node (c) {$\vdots$}
        child [grow=right] {node (q) {$\vdots$} edge from parent[draw=none]
              child [grow=up] {node (s) {$r^2,\, \frac{N}{T^2} \times \frac{N}{T^2}$} edge from parent[draw=none]
                child [grow=up] {node (t) {$r^1,\, \frac{N}{T} \times \frac{N}{T}$} edge from parent[draw=none]
                  child [grow=up] {node (u) {$r^0,\, N \times N$} edge from parent[draw=none]
                    child [grow=up] {node (u) {\begin{tabular}{c}(\#, size)\\of matrices\\ \hline\end{tabular}} edge from parent[draw=none]}
            }          
          }
        }
      }
    }
  }
};
\path (a) -- (j) node [midway] {$\cdots$};
\path (b) -- (g) node [midway] {$\cdots$};
\path (k) -- (l) node [midway] {$\cdots$};
\path (l1) -- (l2) node [midway] {$\cdots$};
\path (l3) -- (l4) node [midway] {$\cdots$};
\path (l5) -- (l6) node [midway] {$\cdots$};
\path (l7) -- (c) node [midway] {$\cdots$};
\path (z) -- (a) node [midway,above] {1};
\path (z) -- (j) node [midway,above] {2};
\path (a) -- (b) node [midway,left] {1};
\path (a) -- (g) node [midway,right] {2};
\path (j) -- (k) node [midway,left] {1};
\path (j) -- (l) node [midway,right] {2};
\end{tikzpicture}}
\caption{The $r$-ary tree $\mathcal{T}_A$ for Strassen's Algorithm ($r = 7$, $T=2$). For $K \times K$ matrices $U$ and $V$, the notation $U_{ij}$ or $(U)_{ij}$ refers to the $(i,j)$th $\frac{K}{T} \times \frac{K}{T}$ block of $U$. Observe that $(U+V)_{ij} = U_{ij} + V_{ij}$.  Each node has children corresponding to the $r$ multiplication expressions $M_i$ (see Figure~\ref{fig:strassen}).  An edge associated with $M_i$ is labeled with the number of terms of $A$ that appear in $M_i$.  Each node $u$ on level $h$, starting with the root as level $0$, corresponds to a matrix that is a weighted sum of $\frac{N}{T^h} \times \frac{N}{T^h}$ blocks of $A$.  The number of blocks of $A$ appearing in such a sum is the product of the edge labels on the path from $u$ to the root of the tree.  For example, $(A_{12}-A_{22})_{12} - (A_{12}-A_{22})_{22} = (A_{12})_{12} - (A_{22})_{12} - (A_{12})_{22} + (A_{22})_{22}$ is a weighted sum of 4 $\frac{N}{T^2} \times \frac{N}{T^2}$ blocks of $A$.  The $N^{\log_T r}$ leaves of $\mathcal{T}_A$ correspond to scalars that are weighted sums of entries of $A$. \cindy{See if there is any value to including the old figure, with updated notation, as part of an appendix.}}
\label{fig:T_A}
\end{figure}
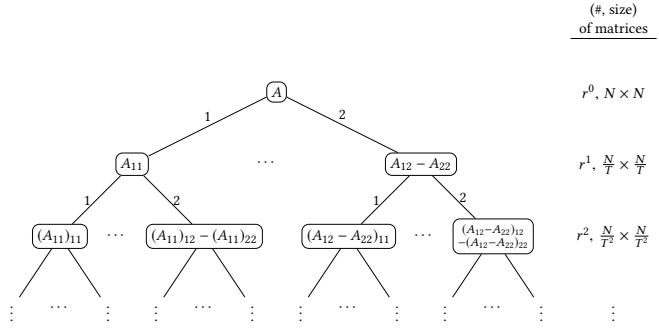

In a conventional PRAM implementation of a fast matrix multiplication algorithm, all the matrices at each level of $\mathcal{T}_A$ and $\mathcal{T}_B$ are computed, and the results are reused.  Since there are $O(\log N)$ levels, we cannot hope to compute all the matrices at each level in a constant-depth circuit.
We give constant-depth threshold circuits that computes all nodes on only a constant number of levels of $\mathcal{T}_A$ and $\mathcal{T}_B$ while using a number of gates arbitrarily close to the total work performed by the fast matrix multiplication algorithm.

Our circuit computes the same $O(N^w)$ scalar products as the underlying fast matrix multiplication algorithm.  These scalars correspond to the leaves of $\mathcal{T}_A$ and those of $\mathcal{T}_B$ respectively.  Our algorithm processes these trees in a top-down fashion to compute the scalars, corresponding to the leaves.  We then appeal to Lemma~\ref{lem:prod-circuit} to compute the product of each scalar corresponding to a leaf of $\mathcal{T}_A$ with the corresponding leaf of $\mathcal{T}_B$.  For both the problems we solve, we will need to consider another tree with similar structure to $\mathcal{T}_A$ and $\mathcal{T}_B$; however, these trees will be used in different ways.

\ojas{Discuss how we may use more general fast matrix multiplication algorithms that do not take \{-1,+1\}-weighted sums.  Lemma 2 allows for more general weights.}

We assume, as in Section~\ref{sec:problem-statement}, that we have a fast matrix multiplication algorithm that multiplies two $T \times T$ matrices using $r$ multiplications.  We describe an improved $TC^0$ circuit for computing the values at the leaves of $\mathcal{T}_A$.  Our results extend naturally to computing the leaves of $\mathcal{T}_B$.  Level $h$ of $\mathcal{T}_A$ contains $r^h$ nodes, each corresponding to an $N/T^h \times N/T^h$ matrix.  Moreover, each entry of each matrix at level $h$ is the $\{-1,1\}$-weighted sum of at most $T^{2h}$ \cindy{This is correct, but a tighter expression would be $\mu^h$, where $\mu$ is the maximum number of entries from matrix $A$ over all expressions $M_i$.  Since this is at most $T^2$, this expression follows.  But it's off by a factor of $2$ for Strassen. Probably more of a comment for the future.} entries of the root matrix, $A$.  Hence if each entry of the integer matrix $A$ requires at most $b$ bits, the number of bits required for each entry of a matrix at level $h$ is at most
\begin{equation}\label{eq:bit-bound} 
\bits((2^b-1)T^{2h}) \leq b + \bits(T^{2h}) = b + O(h \log T).
\end{equation}
For our results we assume $b = O(\log N)$ bits.  Moreover, $T$ is a constant associated with the fast matrix multiplication selected, and $h \leq \log_T N$, hence the scalar values at the leaves of $\mathcal{T}_A$ require $O(\log N)$ bits.

We give a subcubic $TC^0$ circuit for computing the matrix product $C=AB$ in the next section, but first, we illustrate our main ideas by showing how to check $\text{trace}(A^3) \geq \tau$.  As mentioned in Section~\ref{sec:problem-statement}, this allows us, for example, to count triangles in a graph.  The bulk of our work lies in showing how to compute the $N^{\log_T r}$ scalar products prescribed by the fast matrix multiplication algorithm.  Each such scalar product is between a weighted sum of entries of $A$ and a weighted sum of entries of $B$.  We next show how to compute these weighted sums for $A$ with a circuit that computes a constant number of levels of $\mathcal{T}_A$.  An analogous construction works for $B$.

\ojas{Compare different approaches here at a high level and discuss our strategy.}
\subsection{Approach}
Our main goal in the following sections is to give $O(\log \log N)$-depth, $\widetilde{O}(N^{\omega})$-gate and $O(d)$-depth, $\widetilde{O}(N^{\omega + O(\gamma^d)})$-gate threshold circuits for multiplying two $N \times N$ integer matrices with $O(\log N)$-bit entries, for a positive integer $d$ and constant $\gamma < 1$.

We first motivate our approach by attempting to construct a constant-depth and $O(N^{3-\epsilon})$-gate threshold circuit for matrix multiplication using Strassen's algorithm as a guide.  As described in the previous section and Figure~\ref{fig:T_A}, our immediate goal is to compute the scalars associated with the leaves of $\mathcal{T}_A$ (and $\mathcal{T}_B$) for Strassen's algorithm by selecting a constant number of levels of $\mathcal{T}_A$ to compute.  The most natural approach is perhaps to directly compute the leaves, at level $\log_2 N$, without computing any other level.  The leaves of $\mathcal{T}_A$ correspond to $\{-1,1\}$-weighted sums of at most $N$ entries of $A$.  By recalling \eqref{eq:bit-bound} and invoking Lemma~\ref{lem:stronger-sum-circuit}, we can compute each such sum in depth 2 using $O(N \log N)$ gates.  However, we must compute $O(N^{\log_2 7})$ such sums, yielding a total of $\widetilde{O}(N^{1 + \log_2 7}) \approx \widetilde{O}(N^{3.81}) $ gates.  This can be improved to 
$\approx \widetilde{O}(N^{3.58})$ by observing that not all sums have the same number of summands; however, this still fails to achieve our goal.

We can improve the approach by employing addition circuits of depth greater than 2 due to Siu et al.~\cite{siu_depth-size_1991} (Corollary 2).  This allows computation of the desired sums in depth $O(d)$ using $O(d N^{1/d})$ gates, yielding the following result.
\begin{theorem}
\label{thm:poly-bound}
Suppose we are given $N \times N$ integer matrices $A$ and $B$ with entries of size $O(\log N)$ bits.  There is a threshold circuit of depth $O(d)$ that computes the matrix product $AB$ using $\widetilde{O}(d N^{\omega + 1/d})$ gates.
\end{theorem}
We do not include a full proof of this theorem as our main results give superior results, both in terms of gate count and the simplicity of the resulting circuits.  The results of the following sections show how to more carefully select a constant number of levels of $\mathcal{T}_A$ and $\mathcal{T}_B$ in order to improve the exponent in the gate count from $\omega + 1/d$ to $\omega + O(\gamma^d)$ for a constant $\gamma < 1$.

\subsection{Matrix trace}
\label{sec:trace}

We select $t$ levels, $0 = h_0 < h_1 < h_2 < \cdots < h_t$ of the tree $\mathcal{T}_A$.  Our $TC^0$ circuit computes all of the matrices at these $t$ levels of $\mathcal{T}_A$.   Our goal is to compute the scalars corresponding to the leaves of $\mathcal{T}_A$, hence $h_t = \log_T N$.  The benefit of computing level $h_i$ is that each entry of each matrix at level $h_{i+1}$ is then a $\{-1,1\}$-weighted sum of at most $T^{2(h_{i+1}-h_i)}$ matrices at level $h_i$.

Our results rely on parameters associated with our fast matrix multiplication algorithm.  Recall $s_A$ from Definition~\ref{def:d-sparsity}.  We define $\alpha=r/s_A$ and $\beta=s_A/T^2$, and we have that $0 < \alpha \leq 1$ and $\beta \geq 1$ (for Strassen's algorithm, $\alpha = 7/12$ and $\beta = 3$).

\begin{lemma}\label{lem:level-gate-bound}
For $1 \leq i \leq t$, if the matrices at level $h_{i-1}$ of $\mathcal{T}_A$ have been computed, then the matrices at level $h_i$ can be computed in depth 2 using $O((b + h_{i-1})\alpha^{h_{i-1}} \beta^{h_i}N^2)$ gates.
\end{lemma}
\begin{proof}The $r^{h_i}$ nodes at level $h_i$ of $\mathcal{T}_A$ each correspond to an $N/T^{h_i} \times N/T^{h_i}$ matrix.  We set $\delta_i = h_i - h_{i-1}$ for convenience.  We can associate each node $u$ at level $h_i$ with the unique subtree rooted at level $h_{i-1}$ that contains it.  The $N/T^{h_i} \times N/T^{h_i}$ matrix corresponding to $u$ is a sum of at most $T^{2\delta_i}$ blocks of the $N/T^{h_{i-1}} \times N/T^{h_{i-1}}$ matrix associated with the root of the subtree containing $u$. 

We seek a better bound on the number of such blocks we must sum to obtain the matrix associated with $u$.  Let $\text{size}(u)$ represent this quantity, and let $\text{root}(u)$ be the node at level $h_{i-1}$ on the path from $u$ to the root of $\mathcal{T}_A$.   Recall that each edge of $\mathcal{T}_A$ corresponds to one of the fast matrix multiplication expressions $M_i$ and that $a_i$ is the number of distinct blocks of $A$ that appear in $M_i$ (defined in Section~\ref{sec:problem-statement}).  The quantity $\text{size}(u)$ is the product of the $a_i$ associated with the edges on the path from $u$ to $\text{root}(u)$ (see Figure~\ref{fig:T_A}).  Thus for each node $v$ at level $h_{i-1}$, we have:
\ojas{Give a sentence explaining this bound}
\begin{equation}\label{eq:level-bound}
\sum_{\{u \mid \text{root}(u) = v\}} \text{size}(u)
= \sum_{m_1 + \cdots + m_r = \delta_i} {\delta_i \choose m_1,\ldots,m_r} \prod_{1 \leq j \leq r} a_j^{m_j}
= s_A^{\delta_i},
\end{equation}
where the last equality follows from the multinomial theorem.
We now bound the number of gates required to compute the matrices at level $h_i$.  Since we assume the matrices at level $h_{i-1}$ have been computed, by Lemma~\ref{lem:stronger-sum-circuit}, each entry of the matrix associated with node $u$ at level $h_i$ can be computed using $O((b+h_{i-1})\text{size}(u))$ gates in depth 2.  We charge the gate count for $u$ to $\text{root}(u)$, and by \eqref{eq:level-bound} and \eqref{eq:bit-bound}, we have that the number of gates charged to each node at level $h_{i-1}$ is
\begin{equation*}
O((b + h_{i-1}) s_A^{h_{i} - h_{i-1}}N^2/T^{2h_i}), 
\end{equation*}
hence the total number of gates required for level $h_i$ is
\begin{eqnarray*}
O((b + h_{i-1}) r^{h_{i-1}} s_A^{h_{i} - h_{i-1}}N^2/T^{2h_i})
& = & \\
O((b + h_{i-1}) (r/s_A)^{h_{i-1}} (s_A/T^2)^{h_i} N^2),
\end{eqnarray*}
as desired.
\end{proof}

Next we show how to set the $h_i$ so that the number of gates required at each level is approximately balanced.  This yields a total gate count that is, at worst, within a factor of $t$ of the gate count for an optimal setting of the $h_i$.  We must assume the number of multiplications our fast $T \times T$ matrix multiplication algorithm requires is greater than $T^2$.  The results, as stated and proven below, do not hold if we have an optimal fast matrix multiplication algorithm where the number of multiplications, $r = T^2$.  We set $\gamma = \log_\beta(1/\alpha)$. Note that $0 < \gamma < 1$ since $r > T^2$ is equivalent to $\alpha \beta > 1$ (for Strassen's algorithm, $\gamma \approx 0.491$).

\begin{lemma}\label{lem:gate-bound}
Let $h_i = \lceil (1-\gamma^i)\rho \rceil$, for some $\rho > 0$.  Then all the matrices at levels $h_1,\ldots,h_t$ of $\mathcal{T}_A$ can be computed in depth $2t$ using $O(t(\alpha\beta)^\rho (b+ \log N) N^2)$ gates. 
\end{lemma}
\begin{proof}We have $h_i \leq \log_T N$ for all $0 \leq i \leq t$ since the latter is the height of $\mathcal{T}_A$.  By Lemma~\ref{lem:level-gate-bound}, level $h_i$ can be computed in depth 2 using $O((b + \log N)\alpha^{h_{i-1}}\beta^{h_i}N^2)$ gates.

Let $\tilde{h}_i = (1-\gamma^i)\rho$.  Observe that 
\begin{equation*}
\sum_{1 \leq i \leq t} \alpha^{{h}_{i-1}}\beta^{{h}_i} < \beta \sum_{1 \leq i \leq t} \alpha^{\tilde{h}_{i-1}}\beta^{\tilde{h}_i},
\end{equation*}
hence it suffices to bound $\sum_{1 \leq i \leq t} \alpha^{\tilde{h}_{i-1}}\beta^{\tilde{h}_i}$.  The terms in this sum are all equal:  
\begin{align*}
\alpha^{\tilde{h}_{i-1}}\beta^{\tilde{h}_{i}}
= \left(\frac{\alpha \beta}{\alpha^{\gamma^{i-1}}\beta^{\gamma^i}}\right)^\rho
= \left(\frac{\alpha \beta}{\alpha^{\gamma^{i-1}}(\beta^{\gamma})^{\gamma^{i-1}}}\right)^\rho
= \left(\alpha \beta\right)^\rho,
\end{align*}
so that $\sum_{1 \leq i \leq t} \alpha^{{h}_{i-1}}\beta^{{h}_i} = O(t(\alpha\beta)^{\rho})$, from which the claim follows.
\end{proof}

The above lemma establishes a tradeoff in the following sense.  The value $\rho$ impacts the total number of gates. However, we require that $h_t = \log_T N$, which imposes constraints on $t$ and, consequently, the depth of the circuit.  The larger $\rho$, the smaller $t$ needs to be in order for $h_t = \log_T N$.  

The natural strategy of taking $h_i = \lceil i \log_T N / t \rceil$ yields a weaker result, comparable to Theorem~\ref{thm:poly-bound}.
We now establish our main theorems by better quantifying the tradeoff between $\rho$ and $t$.  For these theorems we assume we are given a fast matrix multiplication algorithm and take $\omega = \log_T r$.

\begin{theorem}\label{thm:main2}
Suppose we are given an integer $\tau$ and an $N \times N$ integer matrix $A$ with entries of size $O(\log N)$ bits.  There is a threshold circuit of depth $O(\log \log N)$ that determines whether $\text{trace}(A^3) \geq \tau$ using $\widetilde{O}(N^{\omega})$ gates.
\end{theorem}
\begin{proof}
We appeal to Lemma~\ref{lem:gate-bound}, setting $\rho = \log_T N$.  The gate bound follows from $(\alpha \beta)^\rho = (r/T^2)^{log_T N} = N^{\omega-2}$.  To bound the depth, we set $t = \lfloor \log_{1/\gamma} \log_T N \rfloor + 1 > \log_{1/\gamma} \log_T N$.  This implies:
\begin{align*}
\log_T N - (1-\gamma^t)\log_T N
< \log_T N - (1-1/\log_T N) \log_T N
= 1.
\end{align*}
Thus $h_t = \lceil (1-\gamma^t)\log_T N \rceil = \log_T N$ as desired.

This shows that we can compute the values corresponding to the leaves of $\mathcal{T}_A$ and $\mathcal{T}_B$ in the stated gate and depth bounds.  One may see that each entry $C_{ij}$ is a weighted sum of products, $\sum_{k \in I_{ij}} w_{ijk} p_k$, where each $p_k$ corresponds to a product between a leaf of $\mathcal{T}_A$ and the corresponding leaf of $\mathcal{T}_B$ with each weight $w_{ijk} \in \{-1,1\}$.  We seek to compute 
\begin{eqnarray}
\frac{\text{trace}(A^3)}{2} & = & \sum_{i < j} A_{ij}C_{ij}
=  \sum_{i < j} A_{ij} \left(\sum_{k \in I_{ij}} w_{ijk} p_k\right)\\ \nonumber
& = & \sum_k p_k \left(\sum_{i < j: k \in I_{ij}} w_{ijk}A_{ij}\right). \label{eq:C-sum}
\end{eqnarray}
Thus for each product, $p_k$, we want to multiply it with a $\{-1,1\}$-weighted sum over entries of $A$.  We may compute these weighted sums independently and in parallel with those for $A$ and $B$ using the same techniques.  Thus we seek to compute $N^\omega$ products of 3 $O(\log N)$-bit numbers, and we appeal to Lemma~\ref{lem:prod-circuit} to accomplish this in depth 1 using a total of $\widetilde{O}(N^\omega)$ gates.  A final output gate sums the representations of the products computed by Lemma~\ref{lem:prod-circuit} and compares with the threshold $\tau$.
\end{proof}

We now prove our main theorem by exhibiting a more refined tradeoff between $\rho$ and $t$.  

\begin{theorem}\label{thm:main3}
Suppose we are given an integer $\tau$, an $N \times N$ integer matrix $A$ with entries of size $O(\log N)$ bits, and a positive integer $d$. There is a threshold circuit of depth at most $2d+5$ that determines whether $\text{trace}(A^3) \geq \tau$ using $\widetilde{O}(dN^{\omega + c\gamma^d})$ gates, where $c > 0$ and $\gamma < 1$ are constants with respect to $N$ and $d$ that depend on the parameters of the fast matrix multiplication algorithm employed. 
\end{theorem}
\begin{proof}
As for the previous theorem, we appeal to Lemma~\ref{lem:gate-bound}, this time setting $\rho = \log_T N + \varepsilon \log_{\alpha \beta} N$, for a constant $\varepsilon > 0$ whose value is given below.  We have $(\alpha \beta)^\rho = (r/T^2)^{log_T N}N^\varepsilon = N^{\omega-2 + \varepsilon}$.  

We set $\varepsilon = \gamma^d \log_T(\alpha \beta)/(1-\gamma) > \gamma^d \log_T(\alpha \beta)/(1-\gamma^d)$.  This implies:
\begin{align*}
&\log_T N - (1-\gamma^d)(\log_T N + \varepsilon \log_{\alpha \beta} N)\\
&< \log_T N - (1-\gamma^d)(\log_T N + (\gamma^d/(1-\gamma^d)) \log_T(\alpha \beta) \log_{\alpha \beta} N)\\
&= \log_T N - (1-\gamma^d)\log_T N -\gamma^d \log_T N\\
&=0,
\end{align*}
hence we may take $t < d$ in Lemma~\ref{lem:gate-bound} in order to have $h_t = \log_T N$.  The theorem follows from the argument used in the proof of Theorem~\ref{thm:main2} and taking $c = \log_T(\alpha \beta)/(1-\gamma)$ (for Strassen's algorithm, $c \approx 1.585$) .
\end{proof}

\subsection{Matrix product} 

Now we describe how to compute the entries of the matrix product $C=AB$, where we assume the entries of the $N \times N$ matrices $A$ and $B$ require $O(\log N)$ bits.  We define a tree $\mathcal{T}_{AB}$ with the same structure as $\mathcal{T}_A$ and $\mathcal{T}_B$.  Each node of $\mathcal{T}_{AB}$ represents the product of the matrices of the corresponding nodes of $\mathcal{T}_A$ and $\mathcal{T}_B$.  Hence the root of $\mathcal{T}_{AB}$ represents the matrix $C = AB$, and the leaves represent the $N^{\log_T r}$ scalar products computed by our fast matrix multiplication algorithm.  We compute the root of $\mathcal{T}_{AB}$ in a bottom-up manner assuming that we are only computing the nodes at levels $\log_T N = h_t > h_{t-1} > \ldots > h_1 > h_0 = 0$.  

We let $\alpha_C=r/s_C$ and $\beta_C=s_C/T^2$ be parameters that are a function of the fast matrix multiplication algorithm employed.  Recall that from \eqref{eq:bit-bound} we have that the scalars at the leaves of $\mathcal{T}_{A}$ and $\mathcal{T}_B$ each require $O(\log N)$ bits.  Therefore the products of these scalars represented by the leaves of $\mathcal{T}_{AB}$ also require $O(\log N)$ bits.  

We show that $\mathcal{T}_{AB}$ can be computed in a bottom-up manner with depth and gate bounds comparable to those we obtained for computing the leaves of $\mathcal{T}_A$ and $\mathcal{T}_B$ in the previous section.  We will need a lemma analogous to Lemma~\ref{lem:level-gate-bound}.

\begin{lemma} \label{lem:compute-levels} For $1 \leq i \leq t$, if the matrices at level $h_i$ of $\mathcal{T}_{AB}$ have been computed, then the matrices at level $h_{i-1}$ can be computed in depth 2 using $\widetilde{O}(\alpha_C^{h_{i-1}} \beta_C^{h_i}N^2)$ gates.
\end{lemma}
\noindent Proof: See the appendix.
\vspace*{1em} 

Using the above lemma, the proof of Lemma~\ref{lem:gate-bound} yields the following. 

\begin{lemma}
Let $h_i = \lceil (1-\gamma^i)\rho \rceil$, for some $\rho > 0$.  Then all the matrices at levels $h_1,\ldots,h_t$ of $\mathcal{T}_{AB}$ can be computed in depth $2t$ using $\widetilde{O}(t(\alpha\beta)^\rho N^2)$ gates. 
\end{lemma}

Armed with the above lemmas, we obtain our main results in similar fashion to those for the trace of $A^3$.  The structure of our circuit is that we compute the scalars corresponding to the leaves of $\mathcal{T}_A$ and $\mathcal{T}_B$ as described in the previous section.  We then use Lemma~\ref{lem:prod-circuit} to compute the scalar products between corresponding leaves of $\mathcal{T}_A$ and $\mathcal{T}_B$ in depth 1.  We finally apply the procedure outline above to compute the root of $\mathcal{T}_{AB}$, representing the matrix product $AB$, in a bottom-up manner.  This essentially doubles the depth of the circuits we obtain compared to the corresponding circuits for the trace of $A^3$. 

\begin{theorem}\label{thm:MM-main1}
Suppose we are given $N \times N$ integer matrices $A$ and $B$ with entries of size $O(\log N)$ bits.  There is a threshold circuit of depth $O(\log \log N)$ that computes the matrix product $AB$ using $\widetilde{O}(N^{\omega})$ gates.
\end{theorem}

\begin{theorem}\label{thm:MM-main2}
Suppose we are given $N \times N$ integer matrices $A$ and $B$ with entries of size $O(\log N)$ bits, and a positive integer $d$. There is a threshold circuit of depth at most $4d+1$ that computes the matrix product $AB$ using $\widetilde{O}(dN^{\omega + c\gamma^d})$ gates, where $c > 0$ and $\gamma < 1$ are constants with respect to $N$ and $d$ that depend on the parameters of the fast matrix multiplication algorithm employed. 
\end{theorem}

\section{Matrix-Multiplication Application Background}
\label{sec:MM-apps}
In this section, we provide more background on the relevance of dense matrix multiplication to deep learning.  We also discuss the relevance of matrix multiplication to triangle counting in graphs and triangle counting's relevance to social network analysis.

\paragraph{Deep Learning} As mentioned in Section~\ref{sec:intro}, our primary motivation for neural-circuit-based matrix multiplication is convolutional neural networks for deep learning. 
See Warden's clear explanation of the role of matrix multiplication in convolution steps for neural networks~\cite{Warden2015}, which we summarize here. For more details see the Stanford course notes at \url{http://cs231n.github.io}.  These networks assume the input is a two-dimensional image, with an $n \times n$ grid of pixels, each with $\ell$ channels.  The neural networks usually refer to the number of channels as depth, but in this paper, ``depth'' refers to the number of layers in our circuit.  Typically the number of channels $\ell$ is a constant, but not necessarily just the three classic color channels (red, green, blue).  In a convolutional step, we apply a set of $K$ {\em kernels} to the image. Each kernel looks for a particular subpattern such as a horizontal edge or a splash of red. The kernel considers a small constant $q \times q$ submatrix of pixels (with $\ell$ channels) at a time and is applied across the whole image based on a stride.  This recognizes the pattern no matter where it is in the image. For example, if the stride is four, then the kernel is applied to every fourth column and every fourth row.  A place where the kernel is applied is called a {\em patch}.  For each patch, for each kernel, a dot product scores the extent to which the patch matches the kernel. Computing all the kernels simultaneously is a matrix multiplication.  The first matrix is $P \times Q$, where $P = O(n^2)$ is the number of patches and $Q = q \times q \times \ell$ is the number of elements in a kernel. The second matrix is $Q \times K$.  This gives a $P \times K$ output matrix, giving the score for each patch for each kernel.

Let $N$ be the largest matrix dimension and suppose we use a fast matrix multiplication algorithm that can multiply two $N \times N$ matrices in time $O(N^{\omega})$. Our circuit requires fan-in as large as $O(N^\omega)$.  These are gates at the end that compute the final output matrix entries. Two of the relevant matrix dimensions for convolutional neural networks, $K$ and $Q$, are generally constants.  The third dimension $P$ is not.  However, if the particular architecture can only support fan in $x$, we can break the matrix multiplication into independent pieces, each with at most $\sqrt[\omega]{x}$ rows in the first matrix. These can run in parallel, so they have the same depth, given a large enough architecture.  Thus the unbounded fan-in in our algorithm is not neccesarily a practical limitation for our motivating application.

\paragraph{Social Network Analysis} Social networks of current interest are too large for our circuit methods to be practical for neuromorphic architectures in the near future.  Also social network adjacency matrices are sparse, unlike the dense small matrices for convolutional neural networks we described above.  Nevertheless, we briefly review the motivation for matrix multiplication in this setting. One application is computing the clustering coefficient of an $N$-node graph (or subgraph). The global clustering coefficient is the ratio of the number of triangles in the graph to the number of \emph{wedges} (length-$2$ paths) in the graph.  A degree-$\delta$ node is at the center of  ${\delta \choose 2}$ wedges.  The global clustering coefficient is the fraction of total wedges in the graph that close into triangles. These triangles are common in social networks, where the central node of a wedge may introduce two neighbors.  
Social-network-analysis researchers believe a high global clustering coefficient (also called transitivity) means the graph has community structure. For example, Seshadri, Kolda and Pinar~\cite{BTER-theory-2012} assumed constant global clustering coefficients when proving a structural property of social networks they used for their BTER (Block Two-Level Erd\"{o}s-Renyi) generative model.  Orman, Labatut and Cherifi~\cite{OrmanLC2013} empirically studied the relationship between community structure and clustering coefficient.  They found that high clustering coefficients did imply community structure, although low clustering coefficients did not preclude it.

This paper considered the question: ``Does a graph $G$ have at least $\tau$ triangles?'' The user can pick a value of $\tau$ that represents reasonable community structure for their particular kind of graph. Usually they compute the total number of wedges $D$ in $O(N)$ time and set $\tau$ to some function of $D$ (perhaps just scaling by a constant).

\section{Open problems}
The main open problem is whether we can do matrix multiplication with $\widetilde{O}(N^{\omega})$ gates in constant depth.  Theorem~\ref{thm:main2} shows this can be done in $O(\log \log N)$ depth.  Another open question is lower bounds: What is the minimum depth of a threshold circuit for computing matrix products using $O(N^{3-\varepsilon})$ gates?  Can one show that a constant-depth threshold circuit using $\widetilde{O}(N^\omega)$ gates yields an $O(\log N)$ PRAM algorithm with $O(N^\omega)$ work?  

\ojas{Our technique should extend to other algebraic divide-and-conquer algorithms.  Using the following lemma, one can perform matrix-vector multiplications in $TC^0$}

\ojas{If we do not count inputs as gates, is it possible to beat $O(N^2)$ gates for counting triangles or computing trace?\\}
\ojas{Should extend to other $O(N^\omega)$ divide-and-conquer algebraic algorithms (e.g. inversion).  What else?\\}
\ojas{Better FFT than matrix-vector multiply using Siu et al.  Using layers to exploit structure rather than for gate-efficient addition\\}

One may show that our circuits are $L$-uniform.  Can a stronger uniformity condition be imposed? 

One advantage of neural networks is their low energy relative to CMOS-based electronics. One possible energy model for threshold gates is to charge a gate only if it fires~\cite{uchizawa_energy_2006}.  That is, charge a gate one unit of energy for sending a signal if and only if the weighted sum of the inputs exceeds the threshold.  What is the energy complexity of the kinds of matrix-multiplication circuits we consider?  


\section*{Acknowledgements}
This research was supported by the Laboratory Directed Research and Development program at Sandia National Laboratories, a multi-mission laboratory managed and operated by National Technology and Engineering Solutions of Sandia, LLC., a wholly owned subsidiary of Honeywell International, Inc., for the U.S. Department of Energy's National Nuclear Security Administration under contract DE-NA0003525.

\bibliographystyle{ACM-Reference-Format}
\bibliography{main}


\begin{thebibliography}{26}


\ifx \showCODEN    \undefined \def \showCODEN     #1{\unskip}     \fi
\ifx \showDOI      \undefined \def \showDOI       #1{#1}\fi
\ifx \showISBNx    \undefined \def \showISBNx     #1{\unskip}     \fi
\ifx \showISBNxiii \undefined \def \showISBNxiii  #1{\unskip}     \fi
\ifx \showISSN     \undefined \def \showISSN      #1{\unskip}     \fi
\ifx \showLCCN     \undefined \def \showLCCN      #1{\unskip}     \fi
\ifx \shownote     \undefined \def \shownote      #1{#1}          \fi
\ifx \showarticletitle \undefined \def \showarticletitle #1{#1}   \fi
\ifx \showURL      \undefined \def \showURL       {\relax}        \fi
\providecommand\bibfield[2]{#2}
\providecommand\bibinfo[2]{#2}
\providecommand\natexlab[1]{#1}
\providecommand\showeprint[2][]{arXiv:#2}

\bibitem[\protect\citeauthoryear{Azevedo, Carvalho, Grinberg, Farfel, Ferretti,
  Leite, Jacob~Filho, Lent, and Herculano-Houzel}{Azevedo
  et~al\mbox{.}}{2009}]%
        {azevedo_equal_2009}
\bibfield{author}{\bibinfo{person}{Frederico A.~C. Azevedo},
  \bibinfo{person}{Ludmila R.~B. Carvalho}, \bibinfo{person}{Lea~T. Grinberg},
  \bibinfo{person}{Jos{\'e}~Marcelo Farfel}, \bibinfo{person}{Renata E.~L.
  Ferretti}, \bibinfo{person}{Renata E.~P. Leite}, \bibinfo{person}{Wilson
  Jacob~Filho}, \bibinfo{person}{Roberto Lent}, {and} \bibinfo{person}{Suzana
  Herculano-Houzel}.} \bibinfo{year}{2009}\natexlab{}.
\newblock \showarticletitle{Equal numbers of neuronal and nonneuronal cells
  make the human brain an isometrically scaled-up primate brain}.
\newblock \bibinfo{journal}{\emph{The Journal of Comparative Neurology}}
  \bibinfo{volume}{513}, \bibinfo{number}{5} (\bibinfo{date}{April}
  \bibinfo{year}{2009}), \bibinfo{pages}{532--541}.
\newblock
\showISSN{1096-9861}
\urldef\tempurl%
\url{https://doi.org/10.1002/cne.21974}
\showDOI{\tempurl}


\bibitem[\protect\citeauthoryear{Ballard, Benson, Druinsky, Lipshitz, and
  Schwartz}{Ballard et~al\mbox{.}}{2016}]%
        {ballard_improving_2016}
\bibfield{author}{\bibinfo{person}{Grey Ballard}, \bibinfo{person}{Austin~R.
  Benson}, \bibinfo{person}{Alex Druinsky}, \bibinfo{person}{Benjamin
  Lipshitz}, {and} \bibinfo{person}{Oded Schwartz}.}
  \bibinfo{year}{2016}\natexlab{}.
\newblock \showarticletitle{Improving the {Numerical} {Stability} of {Fast}
  {Matrix} {Multiplication}}.
\newblock \bibinfo{journal}{\emph{SIAM J. Matrix Anal. Appl.}}
  \bibinfo{volume}{37}, \bibinfo{number}{4} (\bibinfo{date}{Jan.}
  \bibinfo{year}{2016}), \bibinfo{pages}{1382--1418}.
\newblock
\showISSN{0895-4798}
\urldef\tempurl%
\url{https://doi.org/10.1137/15M1032168}
\showDOI{\tempurl}


\bibitem[\protect\citeauthoryear{Bini and Lotti}{Bini and Lotti}{1980}]%
        {Bini_Lotti_1980}
\bibfield{author}{\bibinfo{person}{Dario Bini} {and} \bibinfo{person}{Grazia
  Lotti}.} \bibinfo{year}{1980}\natexlab{}.
\newblock \showarticletitle{Stability of fast algorithms for matrix
  multiplication}.
\newblock \bibinfo{journal}{\emph{Numer. Math.}}  \bibinfo{volume}{36}
  (\bibinfo{year}{1980}), \bibinfo{pages}{63--72}.
\newblock


\bibitem[\protect\citeauthoryear{Bl{\"a}ser}{Bl{\"a}ser}{2013}]%
        {blaser_fast_2013}
\bibfield{author}{\bibinfo{person}{Markus Bl{\"a}ser}.}
  \bibinfo{year}{2013}\natexlab{}.
\newblock \bibinfo{booktitle}{\emph{Fast {Matrix} {Multiplication}}}.
\newblock Number~5 in \bibinfo{series}{Graduate {Surveys}}.
  \bibinfo{publisher}{Theory of Computing Library}.
\newblock
\urldef\tempurl%
\url{http://theoryofcomputing.org/articles/gs005/}
\showURL{%
\tempurl}


\bibitem[\protect\citeauthoryear{Davies, Srinivasa, Lin, Chinya, Cao, Choday,
  Dimou, Joshi, Imam, Jain, Liao, Lin, Lines, Liu, Mathaikutty, McCoy, Paul,
  Tse, Venkataramanan, Weng, Wild, Yang, and Wang}{Davies
  et~al\mbox{.}}{2018}]%
        {davies_loihi:_2018}
\bibfield{author}{\bibinfo{person}{M. Davies}, \bibinfo{person}{N. Srinivasa},
  \bibinfo{person}{T.~H. Lin}, \bibinfo{person}{G. Chinya}, \bibinfo{person}{Y.
  Cao}, \bibinfo{person}{S.~H. Choday}, \bibinfo{person}{G. Dimou},
  \bibinfo{person}{P. Joshi}, \bibinfo{person}{N. Imam}, \bibinfo{person}{S.
  Jain}, \bibinfo{person}{Y. Liao}, \bibinfo{person}{C.~K. Lin},
  \bibinfo{person}{A. Lines}, \bibinfo{person}{R. Liu}, \bibinfo{person}{D.
  Mathaikutty}, \bibinfo{person}{S. McCoy}, \bibinfo{person}{A. Paul},
  \bibinfo{person}{J. Tse}, \bibinfo{person}{G. Venkataramanan},
  \bibinfo{person}{Y.~H. Weng}, \bibinfo{person}{A. Wild}, \bibinfo{person}{Y.
  Yang}, {and} \bibinfo{person}{H. Wang}.} \bibinfo{year}{2018}\natexlab{}.
\newblock \showarticletitle{Loihi: {A} {Neuromorphic} {Manycore} {Processor}
  with {On}-{Chip} {Learning}}.
\newblock \bibinfo{journal}{\emph{IEEE Micro}} \bibinfo{volume}{38},
  \bibinfo{number}{1} (\bibinfo{date}{Jan.} \bibinfo{year}{2018}),
  \bibinfo{pages}{82--99}.
\newblock
\showISSN{0272-1732}
\urldef\tempurl%
\url{https://doi.org/10.1109/MM.2018.112130359}
\showDOI{\tempurl}


\bibitem[\protect\citeauthoryear{Esser, Appuswamy, Merolla, Arthur, and
  Modha}{Esser et~al\mbox{.}}{2015}]%
        {esser2015backpropagation}
\bibfield{author}{\bibinfo{person}{Steve~K Esser},
  \bibinfo{person}{Rathinakumar Appuswamy}, \bibinfo{person}{Paul Merolla},
  \bibinfo{person}{John~V Arthur}, {and} \bibinfo{person}{Dharmendra~S Modha}.}
  \bibinfo{year}{2015}\natexlab{}.
\newblock \showarticletitle{Backpropagation for energy-efficient neuromorphic
  computing}. In \bibinfo{booktitle}{\emph{Advances in Neural Information
  Processing Systems}}. \bibinfo{pages}{1117--1125}.
\newblock


\bibitem[\protect\citeauthoryear{Furst, Saxe, and Sipser}{Furst
  et~al\mbox{.}}{1984}]%
        {furst_parity_1984}
\bibfield{author}{\bibinfo{person}{Merrick Furst}, \bibinfo{person}{James~B.
  Saxe}, {and} \bibinfo{person}{Michael Sipser}.}
  \bibinfo{year}{1984}\natexlab{}.
\newblock \showarticletitle{Parity, circuits, and the polynomial-time
  hierarchy}.
\newblock \bibinfo{journal}{\emph{Mathematical systems theory}}
  \bibinfo{volume}{17}, \bibinfo{number}{1} (\bibinfo{date}{Dec.}
  \bibinfo{year}{1984}), \bibinfo{pages}{13--27}.
\newblock
\showISSN{0025-5661, 1433-0490}
\urldef\tempurl%
\url{https://doi.org/10.1007/BF01744431}
\showDOI{\tempurl}


\bibitem[\protect\citeauthoryear{Indiveri, Linares-Barranco, Hamilton, van
  Schaik, Etienne-Cummings, Delbruck, Liu, Dudek, H{\"a}fliger, Renaud,
  Schemmel, Cauwenberghs, Arthur, Hynna, Folowosele, Sa\"{i}ghi,
  Serrano-Gotarredona, Wijekoon, Wang, and Boahen}{Indiveri
  et~al\mbox{.}}{2011}]%
        {indiveri2011neuromorphicMoreAuthors}
\bibfield{author}{\bibinfo{person}{Giacomo Indiveri}, \bibinfo{person}{Bernabe
  Linares-Barranco}, \bibinfo{person}{Tara~Julia Hamilton},
  \bibinfo{person}{Andr{\'e} van Schaik}, \bibinfo{person}{Ralph
  Etienne-Cummings}, \bibinfo{person}{Tobi Delbruck},
  \bibinfo{person}{Shih-Chii Liu}, \bibinfo{person}{Piotr Dudek},
  \bibinfo{person}{Philipp H{\"a}fliger}, \bibinfo{person}{Sylvie Renaud},
  \bibinfo{person}{Johannes Schemmel}, \bibinfo{person}{Gert Cauwenberghs},
  \bibinfo{person}{John Arthur}, \bibinfo{person}{Kai Hynna},
  \bibinfo{person}{Fopefolu Folowosele}, \bibinfo{person}{Sylvain Sa\"{i}ghi},
  \bibinfo{person}{Teresa Serrano-Gotarredona}, \bibinfo{person}{Jayawan
  Wijekoon}, \bibinfo{person}{Yingxue Wang}, {and} \bibinfo{person}{Kwabena
  Boahen}.} \bibinfo{year}{2011}\natexlab{}.
\newblock \showarticletitle{Neuromorphic silicon neuron circuits}.
\newblock \bibinfo{journal}{\emph{Frontiers in Neuroscience}}
  \bibinfo{volume}{5}, \bibinfo{number}{73} (\bibinfo{year}{2011}).
\newblock
\showISSN{1662-453X}
\urldef\tempurl%
\url{https://doi.org/10.3389/fnins.2011.00073}
\showDOI{\tempurl}


\bibitem[\protect\citeauthoryear{Kane and Williams}{Kane and Williams}{2015}]%
        {kane_super-linear_2015}
\bibfield{author}{\bibinfo{person}{Daniel~M. Kane} {and} \bibinfo{person}{Ryan
  Williams}.} \bibinfo{year}{2015}\natexlab{}.
\newblock \showarticletitle{Super-{Linear} {Gate} and {Super}-{Quadratic}
  {Wire} {Lower} {Bounds} for {Depth}-{Two} and {Depth}-{Three} {Threshold}
  {Circuits}}.
\newblock \bibinfo{journal}{\emph{arXiv:1511.07860 [cs]}} (\bibinfo{date}{Nov.}
  \bibinfo{year}{2015}).
\newblock
\urldef\tempurl%
\url{http://arxiv.org/abs/1511.07860}
\showURL{%
\tempurl}
\newblock
\shownote{arXiv: 1511.07860.}


\bibitem[\protect\citeauthoryear{Khan, Lester, Plana, Rast, Jin, Painkras, and
  Furber}{Khan et~al\mbox{.}}{2008}]%
        {khan2008spinnaker}
\bibfield{author}{\bibinfo{person}{Muhammad~Mukaram Khan},
  \bibinfo{person}{David~R Lester}, \bibinfo{person}{Luis~A Plana},
  \bibinfo{person}{A Rast}, \bibinfo{person}{Xin Jin}, \bibinfo{person}{Eustace
  Painkras}, {and} \bibinfo{person}{Stephen~B Furber}.}
  \bibinfo{year}{2008}\natexlab{}.
\newblock \showarticletitle{{SpiNNaker}: mapping neural networks onto a
  massively-parallel chip multiprocessor}. In \bibinfo{booktitle}{\emph{Neural
  Networks, 2008. IJCNN 2008.(IEEE World Congress on Computational
  Intelligence). IEEE International Joint Conference on}}. IEEE,
  \bibinfo{pages}{2849--2856}.
\newblock


\bibitem[\protect\citeauthoryear{Le~Gall}{Le~Gall}{2014}]%
        {le_gall_powers_2014}
\bibfield{author}{\bibinfo{person}{Fran{\c c}ois Le~Gall}.}
  \bibinfo{year}{2014}\natexlab{}.
\newblock \showarticletitle{Powers of {Tensors} and {Fast} {Matrix}
  {Multiplication}}. In \bibinfo{booktitle}{\emph{Proceedings of the 39th
  {International} {Symposium} on {Symbolic} and {Algebraic} {Computation}}}
  \emph{(\bibinfo{series}{{ISSAC} '14})}. \bibinfo{publisher}{ACM},
  \bibinfo{address}{New York, NY, USA}, \bibinfo{pages}{296--303}.
\newblock
\showISBNx{978-1-4503-2501-1}
\urldef\tempurl%
\url{https://doi.org/10.1145/2608628.2608664}
\showDOI{\tempurl}


\bibitem[\protect\citeauthoryear{McCulloch and Pitts}{McCulloch and
  Pitts}{1943}]%
        {mcculloch_logical_1943}
\bibfield{author}{\bibinfo{person}{Warren~S. McCulloch} {and}
  \bibinfo{person}{Walter Pitts}.} \bibinfo{year}{1943}\natexlab{}.
\newblock \showarticletitle{A logical calculus of the ideas immanent in nervous
  activity}.
\newblock \bibinfo{journal}{\emph{The bulletin of mathematical biophysics}}
  \bibinfo{volume}{5}, \bibinfo{number}{4} (\bibinfo{date}{Dec.}
  \bibinfo{year}{1943}), \bibinfo{pages}{115--133}.
\newblock
\showISSN{0007-4985, 1522-9602}
\urldef\tempurl%
\url{https://doi.org/10.1007/BF02478259}
\showDOI{\tempurl}


\bibitem[\protect\citeauthoryear{Merolla, Arthur, Alvarez-Icaza, Cassidy,
  Sawada, Akopyan, Jackson, Imam, Guo, Nakamura, et~al\mbox{.}}{Merolla
  et~al\mbox{.}}{2014}]%
        {merolla2014million}
\bibfield{author}{\bibinfo{person}{Paul~A Merolla}, \bibinfo{person}{John~V
  Arthur}, \bibinfo{person}{Rodrigo Alvarez-Icaza}, \bibinfo{person}{Andrew~S
  Cassidy}, \bibinfo{person}{Jun Sawada}, \bibinfo{person}{Filipp Akopyan},
  \bibinfo{person}{Bryan~L Jackson}, \bibinfo{person}{Nabil Imam},
  \bibinfo{person}{Chen Guo}, \bibinfo{person}{Yutaka Nakamura},
  {et~al\mbox{.}}} \bibinfo{year}{2014}\natexlab{}.
\newblock \showarticletitle{A million spiking-neuron integrated circuit with a
  scalable communication network and interface}.
\newblock \bibinfo{journal}{\emph{Science}} \bibinfo{volume}{345},
  \bibinfo{number}{6197} (\bibinfo{year}{2014}), \bibinfo{pages}{668--673}.
\newblock


\bibitem[\protect\citeauthoryear{Minnick}{Minnick}{1961}]%
        {Min61}
\bibfield{author}{\bibinfo{person}{Robert~C. Minnick}.}
  \bibinfo{year}{1961}\natexlab{}.
\newblock \showarticletitle{Linear-Input Logic}.
\newblock \bibinfo{journal}{\emph{{IRE} Trans. Electronic Computers}}
  \bibinfo{volume}{10}, \bibinfo{number}{1} (\bibinfo{year}{1961}),
  \bibinfo{pages}{6--16}.
\newblock
\urldef\tempurl%
\url{https://doi.org/10.1109/TEC.1961.5219146}
\showDOI{\tempurl}


\bibitem[\protect\citeauthoryear{Muroga}{Muroga}{1959}]%
        {Mur59}
\bibfield{author}{\bibinfo{person}{Saburo Muroga}.}
  \bibinfo{year}{1959}\natexlab{}.
\newblock \showarticletitle{The principle of majority decision logical elements
  and the complexity of their circuits}. In \bibinfo{booktitle}{\emph{{IFIP}
  Congress}}. \bibinfo{pages}{400--406}.
\newblock


\bibitem[\protect\citeauthoryear{Orman, Labatut, and Cherifi}{Orman
  et~al\mbox{.}}{2013}]%
        {OrmanLC2013}
\bibfield{author}{\bibinfo{person}{G\"{u}nce~Keziban Orman},
  \bibinfo{person}{Vincent Labatut}, {and} \bibinfo{person}{Hocine Cherifi}.}
  \bibinfo{year}{2013}\natexlab{}.
\newblock \showarticletitle{An empirical study of the relation between
  community structure and transivity}.
\newblock \bibinfo{journal}{\emph{Studies in Computational Intelligence}}
  \bibinfo{volume}{424} (\bibinfo{year}{2013}), \bibinfo{pages}{99--110}.
\newblock


\bibitem[\protect\citeauthoryear{Schuman, Potok, Patton, Birdwell, Dean, Rose,
  and Plank}{Schuman et~al\mbox{.}}{2017}]%
        {schuman_survey_2017}
\bibfield{author}{\bibinfo{person}{Catherine~D. Schuman},
  \bibinfo{person}{Thomas~E. Potok}, \bibinfo{person}{Robert~M. Patton},
  \bibinfo{person}{J.~Douglas Birdwell}, \bibinfo{person}{Mark~E. Dean},
  \bibinfo{person}{Garrett~S. Rose}, {and} \bibinfo{person}{James~S. Plank}.}
  \bibinfo{year}{2017}\natexlab{}.
\newblock \showarticletitle{A {Survey} of {Neuromorphic} {Computing} and
  {Neural} {Networks} in {Hardware}}.
\newblock  (\bibinfo{date}{May} \bibinfo{year}{2017}).
\newblock
\urldef\tempurl%
\url{https://arxiv.org/abs/1705.06963}
\showURL{%
\tempurl}


\bibitem[\protect\citeauthoryear{Seshadhri, Kolda, and Pinar}{Seshadhri
  et~al\mbox{.}}{2012}]%
        {BTER-theory-2012}
\bibfield{author}{\bibinfo{person}{C Seshadhri}, \bibinfo{person}{Tamara~G.
  Kolda}, {and} \bibinfo{person}{Ali Pinar}.} \bibinfo{year}{2012}\natexlab{}.
\newblock \showarticletitle{Community structure and scale-free collections of
  {E}rd{\"o}s-{R}{\'e}nyi graphs}.
\newblock \bibinfo{journal}{\emph{Physical Review E}} \bibinfo{volume}{85},
  \bibinfo{number}{056109} (\bibinfo{year}{2012}).
\newblock


\bibitem[\protect\citeauthoryear{{\v S}{\'\i}ma and Orponen}{{\v S}{\'\i}ma and
  Orponen}{2003}]%
        {sima_general-purpose_2003}
\bibfield{author}{\bibinfo{person}{Ji{\v r}{\'\i} {\v S}{\'\i}ma} {and}
  \bibinfo{person}{Pekka Orponen}.} \bibinfo{year}{2003}\natexlab{}.
\newblock \showarticletitle{General-{Purpose} {Computation} with {Neural}
  {Networks}: {A} {Survey} of {Complexity} {Theoretic} {Results}}.
\newblock \bibinfo{journal}{\emph{Neural Computation}} \bibinfo{volume}{15},
  \bibinfo{number}{12} (\bibinfo{date}{Dec.} \bibinfo{year}{2003}),
  \bibinfo{pages}{2727--2778}.
\newblock
\showISSN{0899-7667}
\urldef\tempurl%
\url{https://doi.org/10.1162/089976603322518731}
\showDOI{\tempurl}


\bibitem[\protect\citeauthoryear{Siu, Roychowdhury, and Kailath}{Siu
  et~al\mbox{.}}{1991}]%
        {siu_depth-size_1991}
\bibfield{author}{\bibinfo{person}{Kai-Yeung Siu}, \bibinfo{person}{Vwani
  Roychowdhury}, {and} \bibinfo{person}{Thomas Kailath}.}
  \bibinfo{year}{1991}\natexlab{}.
\newblock \showarticletitle{Depth-size tradeoffs for neural computation}.
\newblock \bibinfo{journal}{\emph{IEEE Trans. Comput.}} \bibinfo{volume}{40},
  \bibinfo{number}{12} (\bibinfo{date}{Dec.} \bibinfo{year}{1991}),
  \bibinfo{pages}{1402--1412}.
\newblock
\showISSN{0018-9340}
\urldef\tempurl%
\url{https://doi.org/10.1109/12.106225}
\showDOI{\tempurl}


\bibitem[\protect\citeauthoryear{Siu, Roychowdhury, and Kailath}{Siu
  et~al\mbox{.}}{1995}]%
        {siu_discrete_1995}
\bibfield{author}{\bibinfo{person}{Kai-Yeung Siu}, \bibinfo{person}{Vwani
  Roychowdhury}, {and} \bibinfo{person}{Thomas Kailath}.}
  \bibinfo{year}{1995}\natexlab{}.
\newblock \bibinfo{booktitle}{\emph{Discrete {Neural} {Computation}: {A}
  {Theoretical} {Foundation}}}.
\newblock \bibinfo{publisher}{Prentice-Hall, Inc.}, \bibinfo{address}{Upper
  Saddle River, NJ, USA}.
\newblock
\showISBNx{978-0-13-300708-4}


\bibitem[\protect\citeauthoryear{Strassen}{Strassen}{1969}]%
        {Strassen1969}
\bibfield{author}{\bibinfo{person}{Volker Strassen}.}
  \bibinfo{year}{1969}\natexlab{}.
\newblock \showarticletitle{Gaussian elimination is not optimal}.
\newblock \bibinfo{journal}{\emph{Numer. Math.}} \bibinfo{volume}{13},
  \bibinfo{number}{4} (\bibinfo{date}{Aug} \bibinfo{year}{1969}),
  \bibinfo{pages}{354--356}.
\newblock


\bibitem[\protect\citeauthoryear{Sugiarto, Liu, Davidson, Plana, and
  Furber}{Sugiarto et~al\mbox{.}}{2016}]%
        {sugiarto_high_2016}
\bibfield{author}{\bibinfo{person}{I. Sugiarto}, \bibinfo{person}{G. Liu},
  \bibinfo{person}{S. Davidson}, \bibinfo{person}{L.~A. Plana}, {and}
  \bibinfo{person}{S.~B. Furber}.} \bibinfo{year}{2016}\natexlab{}.
\newblock \showarticletitle{High performance computing on {SpiNNaker}
  neuromorphic platform: {A} case study for energy efficient image processing}.
  In \bibinfo{booktitle}{\emph{2016 {IEEE} 35th {International} {Performance}
  {Computing} and {Communications} {Conference} ({IPCCC})}}.
  \bibinfo{pages}{1--8}.
\newblock
\urldef\tempurl%
\url{https://doi.org/10.1109/PCCC.2016.7820645}
\showDOI{\tempurl}


\bibitem[\protect\citeauthoryear{Uchizawa, Douglas, and Maass}{Uchizawa
  et~al\mbox{.}}{2006}]%
        {uchizawa_energy_2006}
\bibfield{author}{\bibinfo{person}{Kei Uchizawa}, \bibinfo{person}{Rodney
  Douglas}, {and} \bibinfo{person}{Wolfgang Maass}.}
  \bibinfo{year}{2006}\natexlab{}.
\newblock \showarticletitle{Energy {Complexity} and {Entropy} of {Threshold}
  {Circuits}}. In \bibinfo{booktitle}{\emph{Automata, {Languages} and
  {Programming}}} \emph{(\bibinfo{series}{Lecture {Notes} in {Computer}
  {Science}})}. \bibinfo{publisher}{Springer, Berlin, Heidelberg},
  \bibinfo{pages}{631--642}.
\newblock
\showISBNx{978-3-540-35904-3 978-3-540-35905-0}
\urldef\tempurl%
\url{https://doi.org/10.1007/11786986_55}
\showDOI{\tempurl}


\bibitem[\protect\citeauthoryear{Warden}{Warden}{2015}]%
        {Warden2015}
\bibfield{author}{\bibinfo{person}{Pete Warden}.}
  \bibinfo{year}{2015}\natexlab{}.
\newblock \bibinfo{title}{Why {GEMM} is at the heart of deep learning}.
\newblock
  \bibinfo{howpublished}{\url{https://petewarden.com/2015/04/20/why-gemm-is-at-the-heart-of-deep-learning/}}.
\newblock
\newblock
\shownote{online, accessed February 9, 2017.}


\bibitem[\protect\citeauthoryear{Yao}{Yao}{1985}]%
        {yao_separating_1985}
\bibfield{author}{\bibinfo{person}{Andrew C.~C. Yao}.}
  \bibinfo{year}{1985}\natexlab{}.
\newblock \showarticletitle{Separating the polynomial-time hierarchy by
  oracles}. In \bibinfo{booktitle}{\emph{, 26th {Annual} {Symposium} on
  {Foundations} of {Computer} {Science}, 1985}}. \bibinfo{pages}{1--10}.
\newblock
\urldef\tempurl%
\url{https://doi.org/10.1109/SFCS.1985.49}
\showDOI{\tempurl}


\end{thebibliography}

\section*{Appendix}

{\sc \bf Proof of Lemma~\ref{lem:compute-levels}:}

Our proof uses a similar analysis to that of Lemma~\ref{lem:level-gate-bound}.  We need new parameters derived from our fast matrix multiplication algorithm.  For $1 \leq j \leq T^2$, we use $j$ to index the $T^2$ expressions for entries of $C$. We define $c'_j$ as the number of $M_i$ terms that appear in the $j$th expression for an entry of $C$.  For Strassen's algorithm (Figure~\ref{fig:strassen}), we have $c'_1 = 4$, $c'_2 = 2$, $c'_3 = 2$, and $c'_4 = 4$.  Recall the sparsity parameter $s_C$ from Definition~\ref{def:d-sparsity}, and observe that $s_C = \sum_{1 \leq j \leq T^2} c'_j$.

We assume the matrix products at level $h_i$ of $\mathcal{T}_{AB}$ have been computed and compute a node $u$ at level $h_{i-1}$.  We again define $\delta_i = h_i - h_{i-1}$ for convenience.  As an example, suppose we were using the scalars corresponding to the leaves of $\mathcal{T}_{AB}$, at level $h_t$, to construct a $T \times T$ matrix at level $h_t-1$.  Each of the $T^2$ entries of this matrix consists of a $\{-1,1\}$-weighted sum of the scalars at level $h_t$.  

More generally, the matrix corresponding to node $u$ at level $h_{i-1}$ is composed of $T^{2\delta_i}$ blocks that are each $\{-1,1\}$-weighted sums of matrices of size $N/T^{h_i} \times N/T^{h_i}$ from level $h_i$.  We seek to bound the number of terms in each such sum.  Let $u_l$ for $1 \leq l \leq T^{2\delta_i}$ correspond to the blocks of the matrix at node $u$, and let $\text{size}(u_l)$ be the number of terms in the weighted sum of matrices from level $h_i$ that is equal to the block $u_l$.  Using an approach similar to that from the proof of Lemma~\ref{lem:level-gate-bound}, we have:
\begin{eqnarray}
\sum_{1 \leq l \leq T^{2\delta_i}} \text{size}(u_l)  & = &\sum_{m_1 + \cdots + m_{T^2}  = \delta_i} {\delta_i \choose m_1,\ldots,m_{T^2}} \prod_{1 \leq j \leq T^2} (c_j')^{m_j} \nonumber\\ 
& = & \left(\sum_{1 \leq j \leq T^2} c'_j\right)^{\delta_i} = s_C^{\delta_i},
\end{eqnarray}
where, as in the proof of Lemma~\ref{lem:level-gate-bound}, the penultimate equality follows from the multinomial theorem. 
 
Each block $u_l$ is of size $N/T^{h_i} \times N/T^{h_i}$, and by the discussion preceding this lemma and \eqref{eq:bit-bound}, we have that each entry of a block $u_l$ requires $O(\log N)$ bits.  Thus, by Lemma~\ref{lem:stronger-sum-circuit}, we may compute all of the blocks $u_l$ of $u$ in depth 2 with a gate count of:
\begin{align}
& \sum_{1 \leq l \leq T^{2\delta_i}} O(N^2/T^{2h_i} \log N \text{size}(u_l)) \nonumber \\
& = O(N^2/T^{2h_i} \log N \sum_{1 \leq l \leq T^{2\delta_i}} \text{size}(u_l)) \nonumber \\
& = O(N^2/T^{2h_i} \log N s_C^{\delta_i}),
\end{align}
where the last equality follows from the above equation.
Since there are $r^{h_{i-1}}$ nodes in total on level $h_{i-1}$, we may compute all the matrices on level $h_{i-1}$ with a total gate count of,
\begin{eqnarray*}
O(r^{h_{i-1}} s_C^{h_{i} - h_{i-1}}N^2/T^{2h_i} \log N)
&=& O((r/s_C)^{h_{i-1}} (s_C/T^2)^{h_i} N^2 \log N)\\
&=& \widetilde{O}(\alpha_C^{h_{i-1}} \beta_C^{h_i}N^2).
\end{eqnarray*}
\qed

\end{document}